\documentclass[11pt]{article}

\usepackage{amsfonts, caption, subcaption}
\usepackage{amsmath, amsthm}
\usepackage{graphicx}
\usepackage{hyperref}
\usepackage{amssymb}
\usepackage{amsopn}
\usepackage{comment}
\usepackage{color}
\usepackage{tikz}
\usetikzlibrary{positioning,chains,fit,shapes,calc}
\usetikzlibrary{arrows,automata,decorations.pathmorphing}

\usepackage[framemethod=tikz]{mdframed}

\usepackage[a4paper]{geometry}
\newtheorem{theorem}{Theorem}

\newtheorem{claim}{Claim}
\newtheorem{proposition}{Proposition}

\newtheorem{definition}{Definition}

\newcommand{\E}{\mathbb{E}}

\newcommand{\Var}{\text{{\rm{Var}}}}
\newcommand{\Cov}{\text{{\rm{Cov}}}}

\def\State{{Z}}
\def\inter{\mathbf z}
\def\Inter{\mathbf Z}

\def\Dis{{D}}
\def\sign{\mathrm{sign}}
\def\up{\mathbf{d}}
\def\Up{\mathbf{D}}
\def\sort{\mathbf{s}}
\def\Sort{\mathbf{S}}

\def\R{\mathbb R}

\def\eps{\varepsilon}
\def\poly{\mathrm{poly}}

\providecommand\abs[1]{\lvert#1\rvert}
\providecommand\bigabs[1]{\bigl\lvert#1\bigr\rvert}

\DeclareMathOperator{\Ker}{Ker}

\title{Complete Classification \\ of Generalized Santha-Vazirani Sources}

\author{Salman Beigi\footnote{Institute for Research in Fundamental Sciences, Tehran.}
\and Andrej Bogdanov\footnote{Department of Computer Science and Engineering and Institute for Theoretical Science and Communications, Chinese University of Hong Kong.  Supported by HK RGC GRF grants CUHK14208215 and CUHK 14238716.}
\and Omid Etesami\footnote{Institute for Research in Fundamental Sciences, Tehran.}
\and Siyao Guo\footnote{Northeastern University. Part of the work are done while Siyao Guo was a research fellow at the Simons Institute for the Theory of Computing, UC Berkeley.}}

\date{}

\begin{document}

\maketitle

\begin{abstract}
Let $\mathcal{F}$ be a finite alphabet and $\mathcal{D}$ be a finite set of distributions over $\mathcal{F}$.  A Generalized Santha-Vazirani (GSV) source of type $(\mathcal{F}, \mathcal{D})$, introduced by Beigi, Etesami and Gohari (ICALP 2015, SICOMP 2017), is a random sequence $(F_1, \dots, F_n)$ in $\mathcal{F}^n$, where $F_i$ is a sample from some distribution $d \in \mathcal{D}$ whose choice may depend on $F_1, \dots, F_{i-1}$.

We show that all GSV source types $(\mathcal{F}, \mathcal{D})$ fall into one of three categories: (1) non-extractable; (2) extractable with error $n^{-\Theta(1)}$; (3) extractable with error $2^{-\Omega(n)}$.  This rules out other error rates like $1/\log n$ or $2^{-\sqrt{n}}$.

We provide essentially randomness-optimal extraction algorithms for extractable sources.  Our algorithm for category (2) sources extracts with error $\eps$ from $n = \poly(1/\eps)$ samples in time linear in $n$.  Our algorithm for category (3) sources extracts $m$ bits with error $\eps$ from $n = O(m + \log 1/\eps)$ samples in time $\min\{O(nm2^m),n^{O(\abs{\mathcal{F}})}\}$. 

We also give algorithms for classifying a GSV source type $(\mathcal{F}, \mathcal{D})$:  Membership in category (1) can be decided in $\mathrm{NP}$, while membership in category (3) is polynomial-time decidable.
\end{abstract}

\section{Introduction}

Randomness extractors turn a weak source of randomness into almost uniform independent random bits.  One of the first classes of distributions that were considered in the context of randomness extraction are Santha-Vazirani (SV) sources~\cite{SV86}, also called unpredictable-bit sources. An SV~source is a sequence of random bits such that every bit in the sequence has entropy bounded away from zero, even when conditioned on any possible sequence of previous bits.  As already pointed out in~\cite{SV86}, deterministic (seedless) extraction of even a single almost unbiased bit from SV sources is impossible, although these sources have entropy that grows linearly with their length.\footnote{With respect to seeded extraction, a constant seed length is sufficient for all SV sources~\cite{Vadhan12}.}  

In this work we consider deterministic extraction for a natural generalization of Santha-Vazirani sources which was introduced by Beigi, Etesami, and Gohari~\cite{BEG15,BEG17}.  A {\em generalized Santha-Vazirani (GSV) source} is specified by a pair $(\mathcal F, \mathcal D)$, where $\mathcal F$ is a finite set of {\em faces} and $\mathcal D$ is a finite set of {\em dice}, each of which is a probability distribution on $\mathcal F$. (We will assume that each face is assigned positive probability by at least one die.)  A distribution $(F_1, \dots, F_n)$, where the $F_i$s are $\mathcal F$-valued correlated random variables, is admissible by the source if it is generated by the following type of {\em strategy}:  For each $1 \leq i \leq n$, a die $d \in \mathcal D$ is chosen as a function of $F_1, \dots, F_{i-1}$ and $F_i$ is sampled according to the distribution $d$.  

The case $\abs{\mathcal D} = \abs{\mathcal F} = 2$ recovers the definition of SV sources.  In this instance, the dice are two-sided coins, one biased towards heads and the other one towards tails.  

We call a GSV source $(\mathcal F, \mathcal D)$ {\em extractable} with error $\eps$ from $n$ samples if there exists a function $\mathrm{Ext}\colon \mathcal{F}^n \to \{-1, 1\}$ such that for every distribution $(F_1, \dots, F_n)$ in the source, $\abs{\E[\mathrm{Ext}(F_1, \dots, F_n)]} \leq \eps$.  We call a source {\em extractable} if for every error $\eps > 0$ there exists a sample size $n$ for which the source is extractable with these parameters.

Beigi, Etesami and Gohari~\cite{BEG17} showed that randomness extraction from a GSV source is possible assuming the following condition:

\begin{definition}
A GSV source $(\mathcal{F}, \mathcal{D})$ satisfies the {\em Nonzero Kernel Positive Variance ($\mathrm{NK}^+$)} condition if there exists a function $\psi\colon \mathcal{F} \to [-1, 1]$ such that $\E_d[\psi(F)] = 0$ and $\Var_d[\psi(F)] > 0$ for every die $d \in \mathcal{D}$.
\end{definition}

Here, $\E_d$ and $\Var_d$ denote expectation and variance with respect to the distribution of die $d$.  On the other hand, they showed that extractability from such sources necessitates the following {\em Nonzero Kernel} (NK) condition:
\begin{quote}
There exists a nonzero $\psi\colon \mathcal{F} \to [-1, 1]$ such that $\E_d[\psi(F)] = 0$ for every die $d \in \mathcal{D}$.
\end{quote}
In particular, when all faces of all dice have positive probability (an assumption called ``nondegeneracy'' in~\cite{BEG17}), the ($\mathrm{NK}^+$) and (NK) conditions coincide, providing a characterization of extractability for this class of sources.  Their extractor requires $\Theta(1/\eps^3)$ samples to achieve error $\eps$.

There are, however, simple examples of GSV sources (\eqref{eqn:NK-HNK} and~\eqref{eqn:HNK-NKplus} below) that satisfy (NK) but not ($\mathrm{NK}^+$).  The work~\cite{BEG17} does not address the extractability of such sources.

In the setting of GSV sources, the existence of extractors does not appear to easily follow from counting arguments, as is the case of other types of sources for which extraction is known to be possible in principle and the focus is on efficient constructions, such as affine sources~\cite{Bourgain07,Gabizon11}, polynomial sources~\cite{DGW09,Dvir12} and independent blocks~\cite{Bourgain05,CZ16}.

\subsection*{Our Contributions}

Our first contribution is a complete characterization of extractability from GSV sources.  To motivate our result, we first observe that the (NK) condition is, in general, insufficient for extractability.  Consider, for instance the two-diced, three-faced GSV source described by the distributions $d_1 = (0, 0, 1)$ and $d_2 = (\tfrac12, \tfrac12, 0)$.  This source satisfies (NK) with the witness $\psi = (-1, 1, 0)$, but is clearly not extractable as the distribution in which $d_1$ is repeatedly tossed contains no entropy.

A slightly more interesting example is provided by the four-diced, three-faced GSV source 
\begin{equation}
\label{eqn:NK-HNK}
\tag{E1}
   d_1 = (\tfrac12, \tfrac12, 0, 0) 
   \quad d_2 = (0, 0, \tfrac13, \tfrac23) 
   \quad d_3 = (0, 0, \tfrac23, \tfrac13).
\end{equation}
This source also satisfies the (NK) condition (with $\psi = (-1, 1, 0, 0)$).  However, it is not extractable because it contains a ``hidden'' SV source (over two faces):  If die $d_1$ is tossed away and the first two faces are removed, dice $d_2$ and $d_3$ now fail the (NK) condition. 

These two examples suggest the following method for coming up with non-extractable GSV sources:  Start with any source that fails (NK), extend the dice with more faces of zero probability, and add any number of dice that assign positive probability to the new faces.  To describe such sources, we introduce the following natural strengthening of (NK):

\begin{definition}
A GSV source $(\mathcal F, \mathcal D)$ satisfies the {\em Hereditary Nonzero Kernel (HNK)} condition if for all subsets $\mathcal{D}' \subseteq \mathcal{D}$ there exists a nonzero {\em witness} $\psi: \mathcal{F}' \to [-1, 1]$ such that $\E_d[\psi(F)] = 0$ for all $d \in \mathcal{D}'$, where $\mathcal{F}'$ is the set of faces to which at least one die in $\mathcal{D}'$ assigns nonzero probability.
\end{definition}

Clearly (HNK) is a necessary condition for extractability, because if $(\mathcal F, \mathcal D)$ fails (HNK) then $(\mathcal F', \mathcal D')$ fails (NK).  Our first theorem shows that (HNK) is also sufficient.  Moreover, it gives a universal upper bound on the number of samples:

\begin{theorem}
\label{thm:main1}
The following conditions are equivalent for a GSV source $(\mathcal F, \mathcal D)$: 
\begin{enumerate}  
\setlength\itemsep{0pt}
\item $(\mathcal F, \mathcal D)$ satisfies HNK.
\item $(\mathcal F, \mathcal D)$ is extractable.
\item For every $\eps$, $(\mathcal F, \mathcal D)$ is extractable with error $\eps$ from $n = \poly(1/\eps)$ samples in time linear in $n$.
\end{enumerate}
\end{theorem}

In the course of proving Theorem~\ref{thm:main1} we introduce the analytic {\em Mean Variance Ratio (MVR)} condition that turns out to be equivalent to HNK (Proposition~\ref{prop:char}).  We show that a quantitative variant of the MVR condition  determines the best-possible quality of extraction, up to a quadratic gap, even for GSV sources that are not extractable to within arbitrary small error (Propositions~\ref{prop:sufficient} and~\ref{prop:necessary}).

It is natural to ask if $\poly(1/\eps)$ samples are in general necessary for the extractor in part 3 of Theorem~\ref{thm:main1}.  Our second result shows not only that this is the case, but completely characterizes GSV sources that are extractable in a randomness-efficient manner.

\begin{theorem}
\label{thm:main2}
The following conditions are equivalent for a GSV source $(\mathcal F, \mathcal D)$: 
\begin{enumerate}
\setlength\itemsep{0pt}
\item $(\mathcal F, \mathcal D)$ satisfies $\mathit{NK}^+$.
\item For every $\eps$, $(\mathcal F, \mathcal D)$ is extractable with error $\eps$ from $o(1/\eps^2)$ samples.
\item For every $\eps$ and $m$, $(\mathcal F, \mathcal D)$ is extractable with error\footnote{The error of an extractor that outputs multiple bits is the statistical (total variation) distance between its output distribution and the uniform distribution.} $\eps$ and output length $m$ from $n = O(\log(1/\eps) + m)$ samples in time $\min\{O(nm2^m),n^{O(\abs{\mathcal{F}})}\}$.
\end{enumerate}
\end{theorem}

The sample complexity of the extractor in part 3 of Theorem~\ref{thm:main2} is optimal up to the leading constant:  $\Omega(m)$ samples are necessary by entropy considerations, and $\Omega(1/\eps)$ samples are necessary for non-trivial sources\footnote{The exception consists of one-die GSV sources that admit an event of probability exactly half, for which errorless extraction is possible.} by granularity considerations.

Condition $\mathrm{NK}^+$ is strictly stronger than condition HNK.  For example, the source
\begin{equation}
\label{eqn:HNK-NKplus}
\tag{E2}
d_1 = (\tfrac12, \tfrac12, 0, 0), \quad
d_2= (\tfrac14, \tfrac{1}{12}, \tfrac13, \tfrac13),\quad
d_3= (\tfrac{1}{12}, \tfrac14, \tfrac13, \tfrac13).
\end{equation}
satisfies HNK but not $\mathrm{NK}^+$.

Taken together, Theorems~\ref{thm:main1} and~\ref{thm:main2} completely classify non-trivial GSV sources into three categories:  (1) non-extractable, (2) extractable with error $n^{-\Theta(1)}$, and (3) extractable with error $2^{-\Omega(n)}$, where $n$ is the number of samples.  This rules out the existence of GSV sources of other error rates like $1/\log n$ or $2^{-\sqrt{n}}$.

Moreover, sources can be classified algorithmically: Condition HNK can be decided by a $\mathrm{coNP}$ algorithm, while $\mathrm{NK}^+$ is polynomial-time decidable (see Proposition~\ref{prop:nkequiv}).

Figure~\ref{fig:roadmap} indicates the relations between the different conditions for extractability of GSV sources uncovered in this work.

\tikzset{state/.style = {rectangle, rounded corners, draw=black, minimum height=2em, inner sep=2pt, text centered},
         forall/.style = {draw=black, fill=white, minimum height=1.5em, inner sep=0pt, regular polygon, regular polygon sides=3, shape border rotate=180, anchor=east},
         separation/.style = {decorate, decoration={zigzag, amplitude=1.5pt}}}

\vspace{\baselineskip}
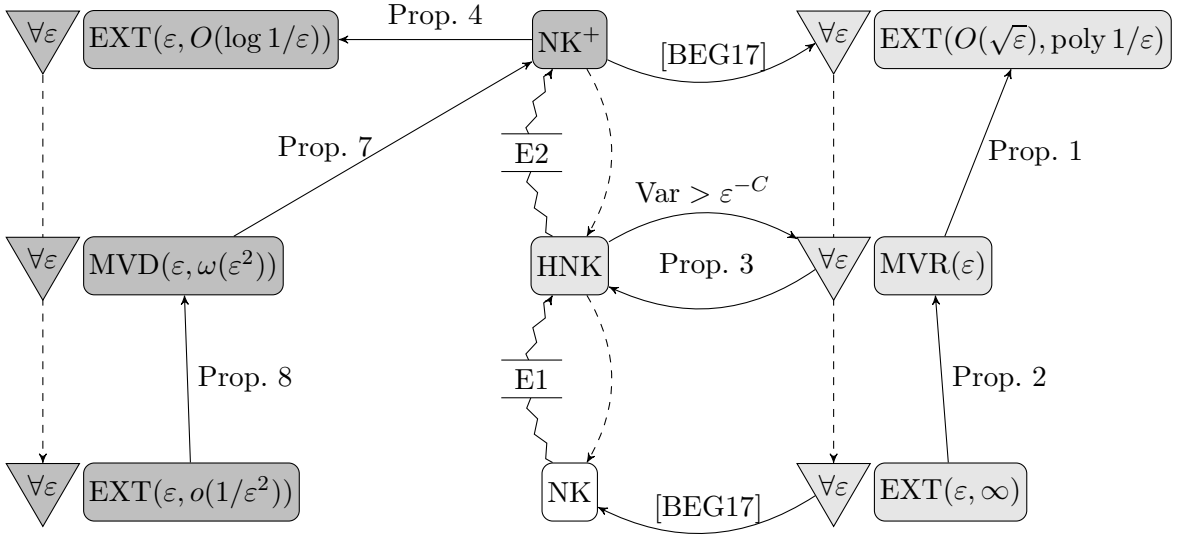
\begin{figure}[h!]
\begin{center}
\begin{tikzpicture}[->,>=stealth']

\node[state, fill=black!25, anchor=west] (EXTlog) at (-1.4, 6) {$\mathrm{EXT}(\eps, O(\log 1/\eps))$};
\node[state, fill=black!25, anchor=west] (MVD) at (-1.4, 3) {$\mathrm{MVD}(\eps, \omega(\eps^2))$};
\node[state, fill=black!25, anchor=west] (EXTsubq) at (-1.4, 0) {$\mathrm{EXT}(\eps, o(1/\eps^2))$};

\node[forall, fill=black!25] (AEXTlog) at (-1.6, 6.1) {$\forall \eps$};
\node[forall, fill=black!25] (AEXTsubq) at (-1.6, 0.1) {$\forall \eps$};
\path (AEXTlog) edge[dashed, ->] (AEXTsubq);
\node[forall, fill=black!25] (AMVD) at (-1.6, 3.1) {$\forall \eps$};

\node[state] (NK) at (5, 0) {NK};
\node[state, fill=black!10] (HNK) at (5, 3) {HNK};
\node[state, fill=black!25] (NKP) at (5, 6) {$\mathrm{NK}^+$};

\node[state, fill=black!10, anchor=west] (EXT) at (9, 0) {$\mathrm{EXT}(\eps, \infty)$};
\node[state, fill=black!10, anchor=west] (MVR) at (9, 3) {$\mathrm{MVR}(\eps)$};
\node[state, fill=black!10, anchor=west] (EXTpoly) at (9, 6) {$\mathrm{EXT}(O(\sqrt\eps), \poly\, 1/\eps)$};

\node[forall, fill=black!10] (AEXT) at (8.8, 0.1) {$\forall \eps$};
\node[forall, fill=black!10] (AEXTpoly) at (8.8, 6.1) {$\forall \eps$};
\path (AEXTpoly) edge[dashed, ->] (AEXT);
\node[forall, fill=black!10] (AMVR) at (8.8, 3.1) {$\forall \eps$};

\path (NKP) edge[->] node[above] {Prop.~\ref{prop:expsmall}} (EXTlog);

\path (NKP) edge[bend left, dashed, ->] (HNK); 
\path (HNK) edge[bend left, dashed, ->] (NK); 

\path (AMVR) edge[bend left, ->] (HNK);
\path (HNK) edge[bend left, ->] node[above] {$\Var > \eps^{-C}$} (AMVR);

\path (MVR) edge[->] node[right] {Prop.~\ref{prop:sufficient}} (EXTpoly);
\path (EXT) edge[->] node[right] {Prop.~\ref{prop:necessary}} (MVR);

\path (AEXT) edge[bend left, ->] node[above] {\cite{BEG17}} (NK);
\path (NKP) edge[bend right, ->] node[above] {\cite{BEG17}} (AEXTpoly);

\node (C) at (6.8, 3) {Prop.~\ref{prop:char}};

\path (MVD) edge[->] node[left] {Prop.~\ref{prop:mvd}} (NKP);

\path (EXTsubq) edge[->] node[right] {Prop.~\ref{prop:lowerquality}} (MVD);

\path (NK) edge[separation, bend left] (HNK);
\fill [white] (4.1, 1.25) rectangle (4.9, 1.75);
\draw [-] (4.1, 1.25) -- (4.9, 1.25);
\draw [-] (4.1, 1.75) -- (4.9, 1.75);
\node (E1) at (4.5, 1.5) {E1};

\path (HNK) edge[separation, bend left] (NKP);
\fill [white] (4.1, 4.25) rectangle (4.9, 4.75);
\draw [-] (4.1, 4.25) -- (4.9, 4.25);
\draw [-] (4.1, 4.75) -- (4.9, 4.75);
\node (E2) at (4.5, 4.5) {E2};
\end{tikzpicture}
\end{center}
\vspace{-\baselineskip}
\caption{A map of our results.  Straight arrows are implications (the dashed ones are immediate) and wiggly arrows are separations.  $\mathrm{EXT}(\eps, n)$ postulates extractability with error $\eps$ from $n$ samples.  Lightly and darkly shaded boxes represent equivalent conditions for extractability and randomness-efficient extractability, respectively.}
\label{fig:roadmap}
\end{figure}

\newpage
\subsection*{Proof Techniques}

Our proofs rely on a combination of probabilistic, algorithmic, and analytic methods.

\paragraph{Feasibility of extraction} 
The extractor of~\cite{BEG17} outputs the sign of $Z_T = \psi(F_1) + \dots + \psi(F_T)$ at the earliest time $T$ when $\abs{Z_T}$ exceeds some pre-specified threshold $M$.  Here, $\psi$ is the witness for condition ($\mathrm{NK}^+$), which ensures that $\E[\psi(F)]$ is always zero and $\Var[\psi(F)]$ is always positive.  Therefore $(Z_t)$ is a martingale with growing variance, and the analysis of~\cite{BEG17} shows that the process terminates by time $n = O(1/\eps^3)$ except with probability $\eps/2$
when $M$ is chosen as $\Theta(1/\eps)$.  Moreover, $Z_T$ must take value in the range $(-(M + 1), -M] \cup [M, M + 1)$, so by the optional stopping time theorem, the bias of $Z_T$ is $\eps/2$ when $M = \Theta(1/\eps)$.

In case only the weaker (HNK) condition holds, $\Var[\psi(F)]$ could be zero for some dice and the value of $Z_t$ may remain constant throughout the process.  On the other hand, (HNK) provides not one but many witnesses $\psi$, one for every subset of the dice.  Proposition~\ref{prop:char} shows how all these witnesses can be combined into a single $\phi\colon \mathcal{F} \to [-1, 1]$ that has positive variance with respect to all the dice, but may have nonzero expectation.  By a careful implementation of this strategy, it is ensured that the ratio $\abs{\E_d[\phi(F)]}/\Var_d[\phi(F)]$ can be made smaller than any pre-specified $\eps > 0$.  This is our Mean Variance Ratio (MVR) condition.  Moreover, $\Var_d[\phi(F)]$ can be lower bounded by $\eps^C$ for some constant $C$ that depends only on the GSV source.

To prove Theorem~\ref{thm:main1} we apply the extractor of~\cite{BEG17} to the function $\phi$.  As $\phi$ may be biased with respect to some dice, $(Z_t)$ may no longer be a martingale, rendering the optional stopping time theorem inapplicable.  In Proposition~\ref{prop:sufficient} we demonstrate that the conclusion of the~\cite{BEG17} analysis still applies in our context.  Intuitively, the (MVR) condition should imply that the 
variance of $Z_t$ grows, and does so at a faster rate than the magnitude of its expectation.  Therefore the stopping time should still be finite, and the component of extraction error incurred by $\abs{\E[Z_T]}$ should be small.  Owing to dependencies between the various steps, a rigorous implementation of these ideas requires substantial care.

\paragraph{Quality and quantity of extracted bits}
For GSV sources that satisfy ($\mathrm{NK}^+$) the extractor of~\cite{BEG17} inherently requires $\Omega(1/\eps)$ samples: On the one hand, to ensure termination with high probability the boundary threshold $M$ can be at most $n$, but on the other hand $Z_T$ may fall anywhere in the range $(-(M + 1), -M] \cup [M, M + 1)$, thereby incurring an error of $\eps = \Omega(1/M)$.\footnote{A tempting alternative is for the ``extractor'' to simply output the sign of $Z_n$ after looking at some predetermined number of samples.  However, this ``extractor'' may incur error $\Omega(1)$ for almost any GSV source.}  To improve the sample complexity, our bit extractor in Theorem~\ref{thm:main2} applies the update rule
\[ Z_{t+1} = Z_t + \frac{\psi(F_t)}{2} \cdot (1 - \abs{Z_t}) \]
and outputs the sign of $Z_n$ for $n = O(\log 1/\eps)$.  Under ($\mathrm{NK}^+$) the sequence $(Z_t)$ is still a martingale, but now the range of $Z_t$ is restricted to the open interval $(-1, 1)$.   On average, the deviation of the step size $Z_{t+1} - Z_t$ conditioned on $Z_t$ is smaller the closer $Z_t$ is to one of the boundary points $\{-1, 1\}$.  We show that the logarithm of $1/(1 - \abs{Z_t})$ grows by a constant on average in every step and apply Azuma's inequality to conclude that $Z_n$ is within $2^{-\Omega(n)}$ of $1$ or $-1$ with probability $1 - 2^{-\Omega(n)}$.  This ensures the bias of the output is inverse exponential in the number of samples.

To extract multiple bits, the state $\Inter_t$ of the above process is extended to encode a probability distribution over $\{0, 1\}^m$.  Initially $\Inter_0$ is the uniform distribution. The distance measure $1 - \abs{Z_t}$ is replaced by a carefully chosen quantity $\Up_{t} \in \R^{2^m}$ which ensures that $\Inter_t$ is a probability distribution that rapidly concentrates on a single entry in $\{0, 1\}^m$, which is the output of the extractor.  Since ($\Inter_t$) is a multi-dimensional martingale, the output must be statistically close to uniform.

\paragraph{Lower bounds} 
Beigi, Etesami, and Gohari~\cite{BEG17} proved that if a source fails the (NK) condition, namely if for all $\psi$ there exists a die $d$ for which $\abs{\E_d[\psi(F)]}/\Var_d[\psi(F)] = \Omega(1)$, then it is not extractable.  In Proposition~\ref{prop:necessary} we prove a quantitatively precise refinement of this statement:  If $\abs{\E_d[\psi(F)]}/\Var_d[\psi(F)] \geq \eps$, then the extraction error must be at least $\Omega(\eps)$.  We conclude that extractability implies the (MVR) condition, which together with a compactness argument (see Proposition~\ref{prop:char}) gives (HNK), proving the ``only if'' direction of Theorem~\ref{thm:main1}.

While this consequence was already established in~\cite{BEG17} by other, combinatorial methods, we obtain a further refinement that is used to prove the ``only if'' direction of Theorem~\ref{thm:main2}.  In Section~\ref{sec:polyexample} we introduce the mean-variance divergence (MVD) condition, which postulates that 
$\abs{\E_d[\psi(F)]} < \eps(\Var_d[\psi(F)] - \delta)$ for all dice. In Proposition~\ref{prop:lowerquality} we show that if MVD fails then extraction with error $\eps$ requires $\Omega(1/\delta)$ samples.  In Proposition~\ref{prop:mvd} we use linear-algebraic duality to show that if ($\mathrm{NK}^+$) fails then so does (MVD) with $\delta = O(\eps^2)$, thereby completing the proof of Theorem~\ref{thm:main2}.

\section{A characterization of extractable GSV sources}
\label{sec:general}

In this Section we prove Theorem~\ref{thm:main1}.  The following analytic condition plays a central role in the proof:

\begin{definition}
A GSV source $(\mathcal{F}, \mathcal{D})$ satisfies the {\em Mean-Variance Ratio condition} with parameter $\eps > 0$ ($MVR(\eps)$) if  there exists a function $\psi:\mathcal F\rightarrow [-1, 1]$ such that for every die $d \in \mathcal{D}$ of a GSV source $(\mathcal F, \mathcal D)$,
\begin{align}
\label{eq:E-V-epsilon}
\tag{MVR}
\big|\E_d[\psi(F)]\big| < \epsilon \Var_d[\psi(F)].
\end{align}
\end{definition}

Proposition~\ref{prop:sufficient} in Section~\ref{sec:feasibility} shows that if a GSV source satisfies $\mathrm{MVR}(\eps)$ then it is extractable with error $O(\sqrt{\eps})$ from $\poly(1/\eps)$ samples.  On the other hand, Proposition~\ref{prop:necessary} in Section~\ref{sec:lowerbound} shows that any GSV source that is extractable with error less than $\eps/10$ (from any number of samples) satisfies $\mathrm{MVR}(\eps)$.  Thus the smallest $\eps$ for which $\mathrm{MVR}(\eps)$ holds measures the best-possible quality of extraction of a GSV source to within a square.

In the case when $\mathrm{MVR}(\eps)$ holds for all $\eps > 0$, the source is extractable.  Surprisingly, proposition~\ref{prop:char} shows that $\forall \eps \mathrm{MVR}(\eps)$ implies HNK.  HNK, in turn, implies a slightly stronger form of $\forall \eps \mathrm{MVR}(\eps)$, which together with Proposition~\ref{prop:sufficient} establishes the extractability of HNK sources from $\eps^{-C}$ samples, where $C$ is a constant that depends only on the source.

\subsection{Feasibility of extraction}
\label{sec:feasibility}

\begin{proposition}
\label{prop:sufficient}
If GSV source $(\mathcal F, \mathcal D)$ satisfies $MVR(\eps)$, then it is extractable from $n$ samples with error at most $3\sqrt{\eps} + 4/\eps v n + O(\eps)$, where $v$ is the minimum of $\Var_d[\psi(F)]$ over all $d \in \mathcal D$.
\end{proposition}

\begin{proof}[Proof of Proposition~\ref{prop:sufficient}] 
Define random variables $X_1, \ldots, X_n$ and $Z_0, \ldots, Z_n$ by $Z_0 = 0$ and $Z_i = Z_{i-1} + X_i$ for $i > 0$, where $X_i = 0$ if $|Z_i| \ge M$, $X_i = \psi({F_i})$ if $|Z_i| < M$,
$F_i$ ($1 \leq i \leq n$) is the $i$-th output of the GSV~source sequence, and $M = 1/\sqrt{\eps}$.  Under this definition, $Z_n$ is uniformly bounded by $M + 1$.  The extractor outputs the sign of $Z_n$.

To prove that the sign of $Z_n$ has small bias, we begin by lower bounding $\Var[Z_n]$.  We will use this lower bound to argue both that the expectation of $Z_n$ in absolute value and that the probability that $Z_n$ remains in the range $(-M, M)$ are small.  These two facts will allow us to conclude that the sign of $Z_n$ is almost unbiased.

\begin{claim}
\label{claim:var-decomp}
$\Var[Z_n] \ge \frac{1}{2} \sum_{i=1}^n \E[\Var[X_i | Z_{i-1}]]$.
\end{claim}
\begin{proof}
By the law of total variance we have 
$$
\Var[Z_i]  = \Var\big[\E[Z_i|Z_{i-1}]\big] + \E \big[\Var[Z_i | Z_{i-1}]\big].
$$
Furthermore, 
\begin{align*}
\Var\big[\E[Z_i|Z_{i-1}]\big] & = \Var\big[Z_{i-1} + \E[X_i | Z_{i-1}]\big] \\
& = \Var[Z_{i-1}] + \Var\big[\E[X_i | Z_{i-1}]\big] + 2 \Cov\big(Z_{i-1} , \E[X_i | Z_{i-1}]\big). \\
\end{align*}
Now we compute
\begin{align*}
\Cov\big(Z_{i-1} , \E[X_i | Z_{i-1}]\big) & = \E\big[\, (Z_{i-1} -\E[Z_{i-1}]) \cdot \E[X_i | Z_{i-1}]\, \big] \\
& \ge - \E\big[\, |Z_{i-1} -\E[Z_{i-1}]| \cdot  |\E[X_i | Z_{i-1}]| \, \big]\\
& \ge - \E\big[\,(2M+2) \cdot \epsilon \Var[X_i | Z_{i-1}]\,\big]\\
& \ge - \E\big[\Var[X_i | Z_{i-1}]/4\big],
\end{align*}
since $|\E[X_i | Z_{i-1}]| \le \epsilon \Var[X_i | Z_{i-1}]$ in both cases $|Z_{i-1}| \ge M$ and $|Z_{i-1}| < M$.
Combining the above three equations and noting that $\Var[\E[X_i|Z_{i-1}]] \ge 0$, we get
$$\Var[Z_i] \ge \Var[Z_{i-1}] - \frac{1}{2}\E\big[\Var[X_i | Z_{i-1}]\big] + \E[\Var[Z_i|Z_{i-1}]].$$
We also have $\Var[Z_i | Z_{i-1}] = \Var[Z_{i-1} + X_i | Z_{i-1}] = \Var[X_i | Z_{i-1}].$
Hence $$\Var[Z_i] \ge \Var[Z_{i-1}] + \frac{1}{2}\E\big[\Var[X_i | Z_{i-1}]\big].$$
The claim now follows by induction on $n$.
\end{proof}

To upper bound $|\E[Z_n]|$, we can write
\begin{equation}
\label{eqn:extexp}
|\E[Z_n]| \le \sum_{i=1}^n \E\bigl[|\E[X_i | Z_{i-1}|\bigr] \leq \sum_{i=1}^n \E\bigl[\eps \Var[X_i | Z_{i-1}|\bigr] \leq 2\eps \Var[Z_n] \leq 2\eps (M+1)^2.
\end{equation}
The first inequality is the triangle inequality.  The second inequality follows from assumption \eqref{eq:E-V-epsilon} when $\abs{Z_i} < M$, and the fact that $\E[X_i | Z_{i-1}]$ is zero otherwise.  The third inequality follows from Claim~\ref{claim:var-decomp}.

Let $p$ be the probability that $|Z_i| < M$ for all $1\leq i\leq n$, i.e., $p=\Pr\big[|Z_i|<M,  1\leq i\leq n\big]$.  Then
\begin{align*}
\Var[Z_n] &\ge \frac{1}{2} \sum_{i=1}^n \E\big[\Var[X_i | Z_{i-1}]\big] \\
&= \frac{1}{2} \sum_{i=1}^n \Pr\big[|Z_{i-1}| < M\big] \cdot  \E\big[\Var[X_i | Z_{i-1}] \,  \big| \, |Z_{i-1}| < M\big] \\
&\ge \frac{1}{2} p n v,
\end{align*}
where the first inequality follows from Claim~\ref{claim:var-decomp}, the second equality follows from the law of conditional expectations, and the third inequality follows because the event $|Z_{i-1}| < M$ contains, in particular the event $|Z_i| < M$ for all $i$ of probability $p$, and conditioned on $\abs{Z_i}$ the conditional variance of $X_i$ is the variance of $\psi(F_i)$.  Therefore,
\begin{equation}
\label{eqn:prob}
 p \le \frac{2 \Var[Z_n]}{n v} \le \frac{2 (M+1)^2}{n v}.
\end{equation}
The bias of the extracted bit is at most
\begin{equation}
\label{eqn:extbias}
 \abs{\Pr[Z_n \geq 0] - \Pr[Z_n < 0]} \leq \abs{p_+ - p_-} + p,
\end{equation}
where $p_+ = \Pr[Z_n \geq M]$ and $p_- = \Pr[Z_n \leq M]$.  To upper bound $\abs{p_+ - p_-}$, we apply the law of conditional expectations to $\E[Z_n]$ to obtain that
\begin{align*}
\bigabs{\E[Z_n] - (p_+ - p_-)M}
  &= \bigabs{p_+ \E[Z_n - M | Z_n \geq M] + p_- \E[Z_n + M | Z_n \leq - M] + p \E[Z_n | \abs{Z_n} < M]} \\
  &\leq p_+ + p_- + pM \\
  &\leq pM + 1.
\end{align*}
By the triangle inequality and \eqref{eqn:extexp},
\[ \abs{p_+ - p_-} \leq \frac{1}{M} \cdot \bigl(\E[Z_n] + pM + 1\bigr) \leq 2\eps\frac{(M+1)^2}{M} + p + \frac{1}{M}. \]
By \eqref{eqn:extbias}, the bias of the extractor is at most $2\eps (M + 3) + 2p + 1/M$.   Assuming, without loss of generality, that $\eps < 1$ and using~\eqref{eqn:prob} we obtain the desired bound for $M = 1/\sqrt{\eps}$.
\end{proof}

\subsection{Impossiblilty of extraction}
\label{sec:lowerbound}

\begin{proposition}\label{prop:necessary}
Let $\eps$ be a sufficiently small constant.  Assume $\mathrm{MVR}(\eps)$ fails for a source $(\mathcal F, \mathcal D)$.  Then $(\mathcal F, \mathcal D)$ is not extractable with error better than $\eps/10$ from any number of samples.
\end{proposition}

\begin{proof}[Proof of Proposition~\ref{prop:necessary}]
Assuming $\mathrm{MVR}(\eps)$ fails we will prove the following claim:

\begin{claim}
\label{claim:induction}
For every $n$, every extractor $\mathrm{Ext}\colon \mathcal{F}^n \to \{0, 1\}$, and every $0 \leq \alpha \leq 1$, if $\E_{A_-}[\mathrm{Ext}] \geq \alpha$ for every strategy $A_-$, then there exists a strategy $A_+$ for which $\E_{A_+}[\mathrm{Ext}] \geq \alpha + (\eps/(1 + \eps)) \cdot \alpha (1 - \alpha)$.
\end{claim}

To derive the theorem from the claim, assume that $\E[\mathrm{Ext}] \geq \alpha = 1/2 - \eps/10$ with respect to every strategy.  By Claim~\ref{claim:induction} there must then exist a strategy for which 
\[ \E[\mathrm{Ext}] \geq \frac12 - \frac{\eps}{10} + \frac{\eps}{1 + \eps} \cdot \frac{1 - \eps^2/100}{4} \]
which is at least $1/2 + \eps/10$.
\end{proof}

\begin{proof}[Proof of Claim~\ref{claim:induction}]
We prove the claim by induction on $n$.  When $n = 0$ the claim holds by checking the cases $\mathrm{Ext} = 0$ and $\mathrm{Ext} = 1$.  We now assume it holds for $n - 1$ and prove it for $n$.  Let $d_-$ be the choice of the first die that minimizes $\E_{A_-}[\mathrm{Ext}]$.  Then
\[ \alpha \leq \E_{d_-}[\alpha(F)], \]
where $\alpha(f)$ is the advantage of $Ext$ conditioned on the first outcome being $f$.  

We now describe the strategy $A_+$.  By $\overline{\mathrm{MVR}(\eps)}$ applied to the function $\psi(f) = \alpha(f) - \alpha$, there exists a die $d_+$ such that 
\begin{equation}
\label{eqn:dplus}
\E_{d_+}[\alpha(F) - \alpha] \geq \eps \Var_{d_+}[\alpha(F)].
\end{equation}
The adversary $A_+$ tosses this die first.  She then plays the strategy that maximizes $\E_{A_+}[\mathrm{Ext}]$ conditioned on the outcome of the first die.  By our inductive assumption, the conditional advantage of $A_+$ when the first outcome is $f$ must be at least $\alpha(f) + (\eps/(1 + \eps)) \cdot \alpha(f) (1 - \alpha(f))$ so that
\[ \E_{A_+}[\mathrm{Ext}] \geq \E_{d_+}\Bigl[\alpha(F) + \frac{\eps}{1 + \eps} \cdot \alpha(F) (1 - \alpha(F))\Bigr]. \]
We can write
\begin{multline}
\label{eq:indcalc}
\E_{d_+}\Bigl[\alpha(F) + \frac{\eps}{1 + \eps} \cdot \alpha(F) (1 - \alpha(F))\Bigr] 
    - \Bigl(\alpha + \frac{\eps}{1 + \eps} \cdot \alpha (1 - \alpha)\Bigr) \\
  = \Bigl(1 + \frac{\eps}{1 + \eps}\Bigr) \E_{d_+}[\alpha(F) - \alpha] 
      - \frac{\eps}{1 + \eps} \Var_{d_+}[\alpha(F)] 
      - \frac{\eps}{1 + \eps} \bigl(\E_{d_+}[\alpha(F)^2] - \alpha^2\bigr).
\end{multline}
We can upper bound the last term by
\[ \bigl(\E_{d_+}[\alpha(F)^2] - \alpha^2\bigr) = \E_{d_+}[\alpha(F) + \alpha] \cdot \E_{d_+}[\alpha(F) - \alpha]
   \leq 2 \E_{d_+}[\alpha(F) - \alpha] \]
since all the $\alpha$s are between zero and one, and the second term is non-negative because by the minimality of $d_-$, $\E_{d_+}[\alpha(F)] \geq \E_{d_-}[\alpha(F)] \geq \alpha$.  We can therefore lower bound the left hand size of \eqref{eq:indcalc} by
\[ \Bigl(1 - \frac{\eps}{1 + \eps}\Bigr) \E_{d_+}[\alpha(F) - \alpha] 
      - \frac{\eps}{1 + \eps} \Var_{d_+}[\alpha(F)] \]
which, by \eqref{eqn:dplus}, must be non-negative.  It follows that the advantage of $A_+$ is at least $\alpha + (\eps/(1 + \eps)) \alpha(1 - \alpha)$, concluding the inductive step.
\end{proof}


\subsection{Proof of Theorem~\ref{thm:main1}}
\label{sec:characterization}

\begin{proposition}
\label{prop:char}
The following conditions are equivalent for a GSV source $(\mathcal{F}, \mathcal{D})$:
\begin{enumerate}
\setlength\itemsep{0pt}
\item For all $\eps > 0$, $(\mathcal{F}, \mathcal{D})$ satisfies $\mathrm{MVR}(\eps)$:  There exists a $\psi: \mathcal{F} \to [-1, 1]$ such that for all dice $d$, $\abs{\E_d[\psi(F)]} < \eps \Var_d[\psi(F)]$.

\item There exists a constant $C$ such that for sufficiently small $\eps > 0$, there exists a $\psi: \mathcal{F} \to [-1, 1]$ such that for all dice $d$, $\abs{\E_d[\psi(F)]} < \eps \Var_d[\psi(F)]$ and $\Var_d[\psi(F)] \geq \eps^C$.

\item $(\mathcal{F}, \mathcal{D})$ satisfies $\mathrm{HNK}$.
\end{enumerate}
\end{proposition}
\begin{proof}
We will show that 1 implies 3 and 3 implies 2.  This will establish equivalence as 2 is a stronger condition than 1.

\medskip\noindent{\it 1 implies 3:} Assume that $(\mathcal{F}, \mathcal{D})$ satisfies $\mathrm{MVR}(\eps)$.  This condition is hereditary, namely if it holds for $(\mathcal{F}, \mathcal{D})$ then it holds for all $(\mathcal{F}', \mathcal{D}')$ in the assumption of HNK.  So in proving 3, we may and will assume, without loss of generality, that $(\mathcal{F}', \mathcal{D'}) = (\mathcal{F}, \mathcal{D})$.  We will moreover assume (by scaling and flipping sign if necessary) that $\psi$ attains the value $1$.

Now consider an infinite decreasing sequence $(\eps_k)$ that converges to zero.  By assumption, for every $k$ there exists a $\psi_k$ such that $\abs{\E_d[\psi_k(F)]} < \eps_k \Var_d[\psi_k(F)]$.  By the pigeonhole principle there must exist a face $f$ for which the set of indices $K = \{k\colon \psi_k(f) = 1\}$ is infinite.  By compactness of $[-1, 1]^\mathcal{F}$ there must exist an infinite subset $K' \subseteq K$ for which the subsequence $\psi_k$ over $k \in K'$ converges to a limit $\psi$.  Then $\psi$ is nonzero as $\psi(f)$ must equal one.  On the other hand, for every $\eps > 0$ there exists a sufficiently large $k \in K'$ such that for every die $d$, 
\[ \abs{\E[\psi_d(F)]} \leq \abs{\E_d[\psi_k]} + \eps \leq \eps\Var_d[\psi_k(F)] + \eps, \]
so $\E_d[\psi_d(F)]$ must equal zero for every $d$.  

\medskip\noindent{\it 3 implies 2:}  The proof is by strong induction on the number of dice $\abs{\mathcal D}$ with $C = 3 \cdot 2^{\abs{\mathcal{D}}} - 3$.  In the base case $\abs{\mathcal D} = 1$, all faces must be assigned nonzero probability by the unique die $d$.  Take any witness $\psi$ for HNK.  Then $\E_d[\psi(F)] = 0$, but $\psi$ must take nonzero value on at least one of the faces, so $\Var_d[\psi(F)] > 0$.  Condition 2 is then satisfied for sufficiently small $\eps > 0$.

For the inductive step, take any $\psi$ that is a witness for HNK with respect to the whole source $(\mathcal{F}, \mathcal{D})$.  Let $\mathcal{D}'$ be the subset of dice $d$ such that $\Var_d[\psi(F)] = 0$ and $v$ be the minimum of $\Var_d[\psi(F)]$ over $d \not\in \mathcal{D}'$.  Then $\mathcal{D}'$ is a proper subset of $\mathcal{D}$ (otherwise, there is a face that is assigned no probability by any die). If $\mathcal{D}'$ is empty, condition 2 follows by the same argument as in the base case.  If not, then by the inductive hypothesis we can choose $\psi'\colon \mathcal{F}' \to [-1, 1]$ such that 
\begin{equation}
\label{eqn:epsd}
\abs{\E_d[\psi'(F)]} < (v\eps^2/8) \cdot \Var_d[\psi(F)] \quad\text{and}\quad \Var_d[\psi(F)] \geq (v\eps^2/8)^{3 \cdot 2^{\abs{\mathcal{D}'}} - 3}.  
\end{equation}

We will show that the function $\phi = \psi + (v\eps/8) \cdot \psi'$ satisfies the conclusion of condition 2.  Here, $\psi'$ is naturally extended as a function on $\mathcal{F}$ by assigning zero on all inputs in $\mathcal{F}\setminus\mathcal{F}'$.  The proof is by cases.

If $d \in \mathcal{D}'$, then $\E_d[\phi(F)] = (v\eps/8)\E_d[\psi'(F)]$, while $\Var_d[\phi(F)] = (v\eps/8)^2\Var_d[\psi'(F)]$.  From these two equalities and \eqref{eqn:epsd} it follows that $\E_d[\phi(F)] < \eps \Var_d[\phi(F)]$.  On the other hand, $\Var_d[\phi(F)] \geq (v\eps/8)^2 \cdot (v\eps^2/8)^{3 \cdot 2^{\abs{\mathcal{D}'}} - 3} \geq \eps^{3 \cdot 2^{\abs{\mathcal{D}}} - 3}$ for sufficiently small $\eps$.

If $d \not\in \mathcal{D'}$, then $\abs{\E_d[\phi(F)]} \leq (v\eps/8) \abs{\E_d[\psi'(F)]} \leq v\eps/8$, while
\begin{align*}
\Var_d[\psi(F)] &\geq \Var_d[\psi'(F)] - 2\abs{\Cov_d[\psi(F), (v \eps/8) \cdot \psi'(F)]} \\
   &= \Var_d[\psi'(F)] - \frac{v\eps}{4} \cdot \abs{\Cov_d[\psi(F), \psi'(F)]} \\
   &\geq \Var_d[\psi'(F)] - \frac{v\eps}{2} \\
   &\geq \frac{v}{2},
\end{align*}
where the last inequality follows from our definition of $v$.  In particular, $\Var_d[\psi(F)] \geq \eps^{3 \cdot 2^{\abs{\mathcal{D}}} - 3}$ for sufficiently small $\eps$.  On the other hand, 
$\abs{\E_d[\psi(F)]} \leq v\eps/8 \leq (\eps/4) \cdot \Var_d[\psi(F)]$, as desired. 
\end{proof}

\begin{proof}[Proof of Theorem~\ref{thm:main1}]
If $(\mathcal F, \mathcal D)$ satisfies HNK, then it also satisfies condition 2 of Proposition~\ref{prop:char}.  By Proposition~\ref{prop:sufficient}, $(\mathcal F, \mathcal D)$ is extractable with error $O(\sqrt{\eps}) + n/\eps^{C + 1}$.  The forward direction follows by setting $n = \eps^{C + 1.5}$.

For the reverse direction, if $(\mathcal F, \mathcal D)$ fails to satisfy HNK, by Proposition~\ref{prop:char}, then it also fails to satisfy $\mathrm{MVR}(\eps)$ for some $\eps > 0$. So by Proposition~\ref{prop:necessary} it is not extractable.
\end{proof}

Alternatively, the reverse direction of Theorem~\ref{thm:main1} can be derived from Theorem 6 of \cite{BEG17} because if $(\mathcal F, \mathcal D)$ fails (NHK) then it contains some $(\mathcal F', \mathcal D')$ which fails (NK).

\section{Randomness-efficient extraction}
\label{sec:improved}

In this Section we prove Theorem~\ref{thm:main2}.  In Section~\ref{sec:single}, we begin with improving the quality of the extractor of~\cite{BEG17} for ($\mathrm{NK}^+$) GSV sources to exponentially small error.  Then in Section~\ref{sec:mult}, we show how to improve the number of extracted bits and prove the implication $1 \rightarrow 3$ in Theorem~\ref{thm:main2}.

In Section~\ref{sec:polyexample} we state and prove a necessary condition for the quality of extraction and use it to prove the remaining implication $2 \rightarrow 1$ in Theorem~\ref{thm:main2}.  

\subsection{An optimal bit extractor} 
\label{sec:single}
\begin{proposition}
	\label{prop:expsmall}
	For every $\eps > 0$, every GSV source that satisfies ($NK^+$) is extractable with error $\eps$ from $O(\log(1/\eps)/v^2)$ samples where $v$ is the minimum of $\Var_d[\psi(F)]$ over all $d \in \mathcal D$.
\end{proposition}
\begin{proof}[Proof of Proposition~\ref{prop:expsmall}]
Define random variables $Z_0, \ldots, Z_n$ by $Z_0 = 0$ and 
$Z_{t+1} = Z_t + (\psi(F_t)/2) \cdot (1 - \abs{Z_t}) $
where
$F_t$ ($1 \leq t \leq n$) is the $t$-th output of the GSV~source sequence. The extractor outputs the sign of $Z_n$.

Under ($\mathrm{NK}^+$) the sequence $(Z_t)$ is still a martingale so that the expectation of $Z_n$ is $0$. But now the range of $Z_t$ is restricted to the open interval $(-1, 1)$.  To prove that the sign of $Z_n$ has small bias, we begin by showing on average, the logarithm of $1/D_t$ grows by a constant on average in every step where $D_t=1-|Z_t|$ is the distance between $Z_t$ and its sign. Then we will use to argue the expectation of $D_n$ is exponentially small. This fact together with $\E[Z_n]=0$ allows us to conclude that the sign of $Z_n$ is exponentially close to unbiased.

\begin{claim} \label{cl:dgrows}
$\E[\ln(1/\Dis_t) - \ln(1/\Dis_{t-1})\ |\ \Dis_1,\dots, \Dis_{t-1}]  \geq v/24.$
\end{claim}
\begin{proof}  Observe that,
$
D_{t}=1-\abs{Z_t}  \leq 1-\sign(Z_{t-1})\cdot Z_t$. By expanding $Z_t$ to   
 $Z_{t-1} + (\psi(F_t)/2)\cdot \Dis_{t-1}$, and replacing $1-\sign(Z_{t-1})\cdot Z_{t-1}$ by $\Dis_{t-1}$, it follows that 
\begin{equation}\label{eq:dis1}
D_t \leq (1- \sign(Z_{t-1}) \cdot\frac{\psi(F_t)}{2})\cdot \Dis_{t-1}.
\end{equation}
Because $|Z_t|\in(-1,1)$ and $D_t>0$, we obtain
	\begin{align*} 
	\E\biggl[\ln\frac{\Dis_{t-1}}{\Dis_t} \ \bigg| \ \Dis_1,\dots, \Dis_{t-1}\biggr]  
	& \geq  \E\bigg[-\ln \Big(1 - \sign(\State_{t-1})\cdot \frac{\psi(F_i)}{2}\Big) \ \bigg| \ \Dis_1,\dots, \Dis_{t-1}\bigg] \\
	& \geq   \E\bigg[ \sign(\State_{t-1})\cdot \frac{\psi(F_i)}{2}+   \frac{1}{6}\Big(\sign(\State_{t-1})\cdot \frac{\psi(F_i)}{2}\Big)^2\bigg]\\
	& =  \E\bigg[ \sign(\State_{t-1})\cdot \frac{\psi(F_i)}{2}\bigg] +   \frac{1}{6}\E\bigg[\Big(\frac{\psi(F_i)}{2}\Big)^2\bigg] \\
	& \geq 0 + v/24. 
	\end{align*}
The second inequality follows from $-\ln(1-x)\geq x+ x^2/6$ and that $\frac{\psi(C_i)}{2}\in(-1/2,1/2)$,
	the third equality follows from the  linearity of expectation, and the last inequality follows from $\E_{d}[\psi(F)]=0$, for all $d\in \mathcal D$, and the  definition of $v$. 		
\end{proof}

 Let us define variable $X_t= \ln(1/\Dis_t) - (vt/24)$, for $t=0,\dots,n$. By Claim~\ref{cl:dgrows}, $\E[X_t|\Dis_0,\dots,\Dis_{t-1}]\geq X_{t-1}$ so that the sequence $(X_t)$ forms a sub-martingale with respect to $\Dis_0,\dots, \Dis_n$.
Moreover, because by triangle inequality, $\Dis_t=1-|z_t|\geq \Dis_{t-1} (1- |\frac{\psi(F_t)}{2}|)\geq \Dis_{t-1}/2$ and by \eqref{eq:dis1}, $\Dis_t\leq \Dis_{t-1} (1+\frac{1}{2})$, $\abs{X_t- X_{t-1}}$ is upper bounded by $c=\ln 2 + v/24$. By Azuma's inequality, for any $k\geq 0$
$$ \Pr[X_n - X_0 \leq - k ] \leq e^{- k^2/2nc^2}.$$
Then by plugging in $k= vn/12$, $X_0=0$ and $c=\ln 2 + v/24$, we obtain $ \Pr[\Dis_n \geq e^{-vn/12}]= 2^{-\Omega(v^2n)}.$ Therefore $\E[D_n]=2^{-\Omega(v^2n)}$. And we can conclude $|\E[\sign{(Z_n)}]|\leq \E[D_n] + |\E[Z_n]|=2^{-\Omega(v^2n)}$ by triangle inequality.
	
\end{proof}

\subsection{Extracting more bits}
\label{sec:mult}
We now explain how the extractor of the previous subsection can be modified to extract multiple bits.  We will prove the following proposition, which is a more detailed restatement of the implication $3 \rightarrow 1$ in Theorem~\ref{thm:main2}. 

\begin{proposition}
	\label{prop:main2}
	For every $\eps > 0$ and $m$, every GSV source that satisfies ($NK^+$) is extractable with error $\eps$ and output length $m$ from $O((\log(1/\eps) + m)/v^2)$ samples where $v$ is the minimum of $\Var_d[\psi(F)]$ over all $d \in \mathcal D$.
\end{proposition}

\paragraph{Extraction procedure.}
Let $M=2^m$. We describe how to output an almost uniform distribution over $[M]$ given ($NK^+$) GSV sources.
Define random variables $\Inter_0, \ldots, \Inter_n$ over $\R^{M}$ by $\Inter_0=\frac{1}{M}\cdot \bold{1}\in\R^{M}$ be the vector all of whose coordinates are $1/M$ and 
\[\Inter_{t}=\Inter_{t-1} +   \frac{\psi(F_t)}{2}\cdot\Up_{t},\]
where $F_t$ ($1 \leq t \leq n$) is the $t$-th output of the GSV~source sequence and $\Up_{t}\in\R^{M}$ is a vector defined below.  The extractor outputs the index of the largest coordinate in $\Inter_n$.

Let $\sort_{i-1}:[M]\rightarrow [M]$ be the function which on input $j\in[M]$ outputs the index of the
$j$-th smallest coordinate in $\inter_{i-1}$ (so that $\inter_{i-1}[\sort_{i-1}(1)] \leq \inter_{i-1}[\sort_{i-1}(2)]\leq \cdots \leq \inter_{i-1}[\sort_{i-1}(M)] $); Then $\up_i$ is given by 
\[\up_i[\sort_{i-1}(j)] = (-1)^{j}\cdot \inter_{i-1}[\sort_{i-1}(j)], \qquad \text{ if } 1\leq j\leq M-1,\]
and $\up_i[\sort_{i-1}(M)] = - \sum_{j=1}^{M-1} \up_i[\sort_{i-1}(j)]$.

 As an example, consider $m=1$.  For each $\inter_t$, if $\inter_{t-1}[1]\leq\inter_{t-1}[2]$, then $\up_t[1]=-\inter_{t-1}[1]$ and $\up_t[2]=\inter_{t-1}[1]$,  otherwise $\up_t[2]=-\inter_{t-1}[2]$ and $\up_t[1]=\inter_{t-1}[2]$. The extractor outputs  $2$ if $\inter_{n}[1]\leq\inter_{n}[2]$ otherwise output $1$. This procedure simulates the extractor in Proposition~\ref{prop:expsmall} by considering variable $Z_t=\Inter_t[2]-\Inter_t[1]$ instead and observing that $Z_t = Z_{t-1} + \frac{\psi(F)}{2}\cdot (1-|Z_t|)$.

\paragraph{Analysis.} Our proof relies on following two claims. The first claim says the values of $\inter_t$ are never negative and add up to one, so at each time they represent a probability distribution over the $M$ possibilities. The second claims says the total mass in this probability distribution gets concentrated in only one of the $M$ possibilities after reading enough symbols from the GSV source.
\begin{claim}\label{cl:prob}
For any $0\leq t\leq n$, $\inter_t[1] +\cdots + \inter_t[M]=1$ and $\inter_t[j]>0$ for any $j\in[M]$
\end{claim}

\begin{claim}\label{cl:single-large} For a sufficiently large constant $C$, and $n\geq C\cdot (m+\log(1/\eps))/v^2$, with probability at least $1-\eps/2$, $\Inter_n[\Sort_n(M)]\geq 1-\eps/2$. 
\end{claim}

We first assume those two claims and prove Proposition~\ref{prop:main2}.
\begin{proof}[Proof of Proposition~\ref{prop:main2}]
For any $j\in [M]$, the sequence of random variables $\Inter_{0}[j],\dots, \Inter_{n}[j]$ forms a martingale because $\E_{d}[\psi(F)]=0$ for any $d\in \mathcal D$. So that $\Inter_n[j]=\Inter_0[j]=1/M$.  By Claim~\ref{cl:prob}, $\inter_n$ is a probability distribution. Consider a modification $\mathrm{Ext}'$ of our extractor which outputs $j$ with probability $\Inter_{n}[j]$ for any $j\in[M]$.  Then the output of $\mathrm{Ext}'$ is uniformly distributed over $[M]$, because for any $j\in[M]$, $\mathrm{Ext}'$ outputs $j$ with probability $\E[\Inter_{n}[j]]=1/M$. 

Conditioned on $\Inter_n$, the output of our extractor is different from $\mathrm{Ext}'$ with probability at most $1-\Inter_n[\Sort_n(M)]$. Thus the error of our extractor is at most $1- \E\big[\Inter_n[\Sort_n(M)]\big]$.
By Claim~\ref{cl:single-large} and $\Inter_n\geq 0$,  for $n=O((m+\log(1/\eps^2)/v^2))$,  $\E\big[\Inter_n[\Sort_n(M)]\geq (1-\eps/2)^2\geq 1-\eps$. Proposition~\ref{prop:main2} follows.
\end{proof}
Now we prove Claim~\ref{cl:prob} and~\ref{cl:single-large}.
\begin{proof}[Proof of Claim~\ref{cl:prob}]
We prove by induction on $t$. The base case holds by the definition of $\inter_0$. Moreover, assuming that coordinates of $\inter_{t-1}$ sum to one, we find that $ \inter_t[1] +\cdots + \inter_t[M]=1$ due to $\sum_{j=1}^{M}\up_{t}[j]=0$.  Also, $\inter_{t}>0$ follows from $\inter_{t-1}>0$, $|\psi|\leq 1$ and  that $|\up_t[j]|\leq \inter_{t-1}[j]$ for all $j\in[M]$. The latter inequality, when written in the form $|\up_t[\sort_{t-1}(j)]|\leq \inter_{t-1}[\sort_{t-1}(j)]$ is easy to verify: It is immediate for $1\leq j\leq M-1$ and for $j=M$ we have 
$$\big|\up_t[\sort_{i-1}(M)]\big|   = \bigg|\sum_{j=1}^{M-1}(-1)^{j}\cdot\inter_{t-1}[\sort_{t-1}(j)]\bigg|  
\leq \inter_{t-1}[\sort_{t-1}(M-1)]\leq \inter_{t-1}[\sort_{t-1}(M)],$$ 
where the middle inequality comes from the fact that $\inter_{t-1}[\sort_{t-1}(j)]$'s have been sorted and $\inter_{t-1}> 0$ so that we can rewrite $\sum_{j=1}^{M-1}(-1)^{j+1}\cdot\inter_{i-1}[\sort_{i-1}(j)]$ and show  $0<\sum_{j=1}^{M-1}(-1)^{j+1}\cdot\inter_{i-1}[\sort_{i-1}(j)]\leq \inter_{t-1}[\sort_{t-1}(M-1)]$ as follows.
\[\inter_{t-1}[\sort_{t-1}(1)] + \sum_{j=2}^{M/2}\big(\inter_{t-1}[\sort_{t-1}(2j-1)]-\inter_{t-1}[\sort_{t-1}(2j-2)]\big)>0
,\]
\[\inter_{t-1}[\sort_{t-1}(M-1)] -  \sum_{j=2}^{M/2}\big(\inter_{t-1}[\sort_{t-1}(2j-2)]-\inter_{t-1}[\sort_{t-1}(2j-3)]\big)\leq \inter_{t-1}[\sort_{t-1}(M-1)].\]
\end{proof}

\begin{proof}[Proof of Claim~\ref{cl:single-large}]
To prove $\Pr[\Inter_n[\Sort_n(M)]\geq 1-\eps/2]\geq 1-\eps/2$, it is sufficient to show,
\begin{equation}
\label{eq:goal}
\Pr\Big[\exists j_0 \in [M], \text{ s.t. } \forall j\neq j_0, \Inter_{n}[j]\leq \frac{\eps}{2M}\Big] \geq 1-\eps/2.
\end{equation}
To show this, for any $j\in[M]$,  we define the sequence $\mathbf{X}_0[j],\dots, \mathbf{X}_n[j]$ where 
$\mathbf{X}_t[j]=\ln{(1/\Inter_{t}[j])},$ 
and  the sequence $\mathbf{Y}_0[j],\dots, \mathbf{Y}_n[j]$ where 
$\mathbf{Y}_0[j]= 0$ and for $t\geq 1$
$\mathbf{Y}_t[j] = \mathbf{Y}_{t-1}[j] + \mathbf{X}_t[j] - \E\big[\mathbf{X}_t[j]\, \big| \,\mathbf{X}_{t-1}[j],\dots, \mathbf{X}_0[j]\big].$  

Observe that the sequence of $\mathbf{Y}_0[j],\dots,\mathbf{Y}_n[j]$ forms a martingale. In addition, for any $j\in[M]$ and $1\leq t\leq n$, $|\mathbf{Y}_t[j]-\mathbf{Y}_{t-1}[j]|\leq\ln3$, because 
 \[\mathbf{Y}_t[j]-\mathbf{Y}_{t-1}[j]=\mathbf{X}_t[j] - \mathbf{X}_{t-1}[j]+ \E\big[\mathbf{X}_{t-1}[j]-\mathbf{X}_t[j]\, \big| \,\mathbf{X}_{t-1}[j],\dots, \mathbf{X}_0[j]\big]\]
and $\mathbf{X}_t[j] - \mathbf{X}_{t-1}[j]=\ln \Big(1+ \frac{\Up_{t}[j]}{\Inter_{t-1}[j]}\cdot \frac{\psi(F_t)}{2}\Big)\in [\ln\frac{1}{2},\ln\frac{3}{2}]$ ($|\Up_{t}[j]|\leq |\Inter_{t-1}[j]|$ has been established). 

By Azuma's inequality, it holds that for any $k$, 
$ \Pr[ |\mathbf{Y}_n[j]| >k ] \leq e^{-k^2/2n(\ln 3)^2}.$ By union bound, it holds that 
\begin{align}\label{eq:mart}
\Pr[\forall j, |\mathbf{Y}_n[j]|\leq k ]\geq 1- M\cdot e^{-k^2/2n(\ln 3)^2}.
\end{align} 
We claim that there exists $j_0\in[M]$ such that for any $j\neq j_0$,
\begin{equation}\label{eq:acu}
\mathbf{X}_n[j] - \mathbf{Y}_n[j] \geq \frac{nv}{48}.
\end{equation}
Then (\ref{eq:mart}) and $(\ref{eq:acu})$ imply 
\[\Pr[\exists j_0\in[M], \text{ s.t. }\forall j\neq j_0, \mathbf{X}_n[j]\geq \frac{nv}{48}-t] \geq 1- M\cdot e^{-k^2/2n(\ln 3)^2}.\]
Plugging  in $k=nv/96$, we obtain (\ref{eq:goal}) for $n\geq C (m+ \log(1/\eps))$ where $C\geq \max(96, 192(\ln3)^2/v^2) \cdot  \ln2$. It remains to prove (\ref{eq:acu}). By expanding the recursive relation of $\mathbf{Y}_t$, for any $t\geq 1$, we have
\begin{equation}\label{eq:dxy}
\mathbf{Y}_t[j] - \mathbf{X}_t[j] = \sum_{k=0}^{t-1} \Big(\mathbf{X}_k[j] - \E\big[\mathbf{X}_{k+1}[j]\,\big|\, \mathbf{X}_0[j],\dots, \mathbf{X}_k[j]\big]\Big).
\end{equation}
Intuitively, $\mathbf{Y}_t[j]- \mathbf{X}_t[j]$ accumulates the shifts of $\mathbf{X}_0,\dots, \mathbf{X}_{t-1}$ from being a martingale. Now for any $k$ we compute  
\begin{align*}
\mathbf{X}_k[j] - \E\big[\mathbf{X}_{k+1}[j] \,\big|\,\mathbf{X}_0[j],\dots, \mathbf{X}_k[j]\big]  
& = \E\bigg[\ln \frac{\Inter_{k+1}[j]}{\Inter_{k}[j]}\,\bigg |\, \mathbf{X}_0[j],\dots, \mathbf{X}_k[j]\bigg] \\
& = \E\bigg[\ln \Big(1+ \frac{\Up_{k+1}[j]}{\Inter_{k}[j]}\cdot \frac{\psi(F_k)}{2}\Big)\, \bigg|\, \mathbf{X}_0[j],\dots, \mathbf{X}_k[j]\bigg]\\
& \leq - \frac{1}{6}\E\Big[\Big(\frac{\Up_{k+1}[j]}{\Inter_{k}[j]}\cdot \frac{\psi(F_k)}{2}\Big)^2 \,\Big |\, \mathbf{X}_0[j],\dots, \mathbf{X}_k[j]\Big],
\end{align*}
where the second equation is by $\Inter_{k+1}[j]=\Inter_{k}[j]+\Up_{k+1}[j]\cdot \frac{\psi(C_k)}{2}$ and the third equation is because, by Taylor expansion, $\ln(1+x)=\sum_{\ell=1}^{\infty} \frac{(-1)^{\ell+1}}{\ell !} x^{\ell}$ and $\ln(1+x)\leq -\frac{x^2}{6}$ for $x=\frac{\Up_{k+1}[j]}{\Inter_{k}[j]}\cdot \frac{\psi(C_k)}{2}\in(-1/2,1/2)$ (note that we have already established $|\Up_{k+1}[j]|\leq |\Inter_{k}[j]|$). Furthermore, for any $k$ and $j\neq \sort_k[M]$, because $|\up_{k+1}[j]|=|\inter_{k}[j]|$, we find that   
\begin{equation}
\label{eq:shift}
\mathbf{X}_k[j] - \E\big[\mathbf{X}_{k+1}[j] \,\big|\,\mathbf{X}_0[j],\dots, \mathbf{X}_k[j]\big] \leq - \frac{1}{6}\E\Big[\Big( \frac{\psi(F_k)}{2}\Big)^2 \,\Big |\, \mathbf{X}_0[j],\dots, \mathbf{X}_k[j]\Big]  \leq -\frac{v}{24},
\end{equation}
Because for every $k$, there exists at most a single $j$ such that $|\up_{k+1}[j]|\neq|\inter_{k}[j]|$. By averaging argument, there exists at most a single $j_0$, such that  $|\up_{k+1}[j]|\neq|\inter_{k}[j]|$ happens at least $n/2$ times for $k=0,\dots,n$. For other $j\neq j_0$, $|\up_{k+1}[j]|=|\inter_{k}[j]|$ happens at least $n/2$ times.  Thus by (\ref{eq:shift}),  (\ref{eq:dxy}), for any $j\neq j_0$, 
\[\mathbf{Y}_n[j]-\mathbf{X}_n[j]\leq\sum_{k:|\up_{k+1}[j]|=|\inter_{k}[j]|} (-\frac{v}{24}) + \sum_{k:|\up_{k+1}[j]|\neq |\inter_{k}[j]|} 0 \leq  -\frac{v}{24}\cdot\frac{n}{2}=-\frac{nv}{48}.\] 
The desired conclusion follows.
\end{proof}

\paragraph{A more efficient implementation.}

A straightforward implementation of the extractor in Proposition~\ref{prop:main2} is simultaneous manipulation of $2^m$ martingales, which requires time $nm2^m$. To extract $\omega(\log{n})$ number of bits, the extractor will run in super-polynomial time which is inefficient.

We give a more efficient implementation of our extractor which orchestrates all the $2^m$ martingale updates simultaneously in time $\min\{nm2^m, n^{O(|\mathcal F|)}\}$. When the number of dice $|\mathcal F|$ is a constant, the extractor runs in polynomial time even for extracting linear number of bits.

For every $0\leq t\leq n$, we use two lists $(\mathrm{L_t}, \mathrm{LL_t})$ to keep track of martingales in $\Inter_{t}$. $\mathrm{LL_t}$ keeps track of martingales who have been one of the largest martingales in last $t$ steps. And $\mathrm{L_t}$ keeps necessary information for other martingales. 
Every element in $\mathrm{LL_t}$ is a pair $(v,t_v)$ which represents a martingale with value $v$ whose first time being the largest martingale is in $\Inter_{t_v}$. Every element in $\mathrm{L_t}$ is a pair of $(v, c(v))$ where $v$ is a value and  $c(v)$ is the number of martingales in $\Inter_t$ (but not in $\mathrm{LL_t}$) with value $v$. In particular, we define $\mathrm{LL_0}=\emptyset$ and $\mathrm{L_0}=\{(1/2^m,2^m)\}$. 

The main observation is that every martingale in $\mathrm{L_t}$ has been updated in a multiplicative way where the multiplicative factor comes from a fixed set  $\{1\pm\psi(f)/2: f\in\mathcal F\}$. Thus the size $\mathrm{L_t}$ is at most $\binom{t+2|\mathcal F|}{2|\mathcal F|}=n^{O(|\mathcal F|)}$. Moreover, there are at most $t$ martingales in $\mathrm{LL_t}$ and at most $2^m$ martaingales in total. Therefore we will operate on at most $\min(2^m, n^{O(|\mathcal F|)}+n)$ objects in every step. In particular,  sorting values among $\mathrm{LL_{t-1}}$ and $\mathrm{L_{t-1}}$ runs in $\min(m2^m, n^{O(|\mathcal F|)})$ and after that, the order of $j$th martingale with value $v$ in $\mathrm{LL_{t-1}}$ or a martingale in $\mathrm{L_{t-1}}$ can be obtained in time  $\min(2^m, n^{O(|\mathcal F|)})$.

Given $\mathrm{LL_{t-1}},\mathrm{L_{t-1}}$, it is sufficient to continue the update rule to obtain $\mathrm{LL_{t}},\mathrm{L_{t}}$.
We go over $(v,c(v))\in \mathrm{L_{t-1}}$ in order (increasing in $v$) and we derive at most two groups of martingales with values $v_1=v\cdot(1+\psi(f)/2)$ and $v_2=v\cdot(1-\psi(f)/2)$. Moreover, knowing the order in $Z_{t-1}$ allows us to know the size of each group and to update  $(v_1,c(v_1))$ and $(v_2,c(v_2))$ in $\mathrm{L_{t}}$ accordingly. Updating values in $\mathrm{LL_{t-1}}$ is straightforward. In the end, if the largest value $v$ is in $\mathrm{L_t}$ (if several martingales have the same value, take the one with highest order in step $t-1$), then we move the martingale from $\mathrm{LL_t}$ to $\mathrm{L_t}$ by adding $(v,t)$ into $\mathrm{LL_t}$ and updating $(v,c(v))$ to $(v,c(v)-1)$ in $\mathrm{L_t}$.
 
From $\mathrm{LL_n}$, we obtain the largest martingale in $\Inter_n$. In order to track back its identity, for $t=0$, we define the order of all martingales by their indexes and for every $1\leq t\leq n$, we define the order martingales with the same value in $\mathrm{LL_t}$ by their orders in $\Inter_{t-1}$.  We prove by strong induction that given the order of a martingale in $\Inter_{t}$ for $0\leq t\leq n$, we can track back its order in $\Inter_{0}$ which is its identify.

The base case is true for $t=0$. Suppose it holds for $t'\leq t-1$.
Given a martingale in $\mathrm{LL_{t}}$ with value $v$,  if $t_v\leq t-1$, we can apply induction hypothesis to track back the largest one in $\Inter_{t_v}$. If $t_v = t$, then we run our updating procedures on  $\mathrm{LL_{t-1}}$ to identify the last martingale in $\mathrm{L_{t}}$ whose value becomes $v$.  Similarly, for any $j$, given the $j$th martingale with value $v$ in $\mathrm{L_{t}}$, we run the updating procedures on $\mathrm{L_{t-1}}$ to identify the $j$th martingale whose value becomes $v$ in $\mathrm{L_{t}}$. Given $\mathrm{L_{t-1}}$ and $\mathrm{LL_{t-1}}$, knowing its order among martingales with the same value is sufficient to identify its order among all martingales in $\Inter_{t-1}$. So we can apply induction hypothesis to track back its identity in $Z_0$.

\subsection{A lower bound on the quality of extraction}
\label{sec:polyexample}

The implication $2 \rightarrow 1$ in Theorem~\ref{thm:main2} follows readily from Propositions~\ref{prop:mvd} and \ref{prop:lowerquality} below.  These refer to an analytic condition that characterizes randomness-efficient extractability called the mean-variance divergence (MVD) condition, which can be viewed as the suitable analogue of the MVR condition in Section~\ref{sec:general}.

The {\em kernel} of GSV source $(\mathcal F, \mathcal D)$, denoted by $\Ker \mathcal{D}$, is the set of all $\psi\colon \mathcal{F} \to \R$ such that $\E_d[\psi_d(F)] = 0$ for all dice $d \in \mathcal{D}$.

\begin{proposition}
\label{prop:nkequiv}
A GSV source $(\mathcal F, \mathcal D)$ satisfies ($\mathit{NK}^+$) if and only if for every die $d \in \mathcal{D}$ there exists a function $\psi_d \in \Ker \mathcal{D}$ that is not constant on the support of $d$.
\end{proposition}
\begin{proof}
The forward direction follows by setting all $\psi_d$ to equal the witness $\psi$ for the ($\mathit{NK}^+$) condition.

For the reverse direction, let $\psi = \sum_{d \in \mathcal{D}} N_d \psi_d$ where $N_d$ are independent random variables, each uniformly distributed over some finite set $\mathcal{N} \subseteq \R$ of size more than $\abs{D}$.  By linearity, $\psi$ is in $\Ker \mathcal{D}$.  Moreover, for each die $d$ and each possible choice of the values $N_{d'}$ for $d' \neq d$, the sum $\sum N_d \psi_d$ can be constant on the support of $d$ for at most one choice of $N_d$ (for if two such choices existed then $\psi_d$ itself must be constant on the support of $d$).  Therefore, $\psi$ is constant on $d$ with probability at most $1/\abs{\mathcal{N}}$.  Since $\abs{\mathcal{N}} > \abs{\mathcal{D}}$, the existence of an ($\mathit{NK}^+$) witness $\psi$ follows from the union bound.
\end{proof}

\begin{claim}
\label{claim:nkdual}
If ($\mathit{NK}^+$) fails for GSV source $(\mathcal F, \mathcal D)$ then there exists a die $d \in \mathcal{D}$ such that for every pair of faces $f^*, f_*$ in the support of $\mathcal{D}$ there exists a function $\beta\colon \mathcal{D} \to \R$ such that for all functions $\psi\colon \mathcal{F} \to \R$,
\begin{equation}
\label{eq:nkpdual}
\psi(f^*) - \psi(f_*) = \sum_{d' \in \mathcal{D}} \beta(d') \cdot \E_{d'}[\psi(F)].
\end{equation}
\end{claim}
\begin{proof}
If $f^* = f_*$ the conclusion holds with $\beta = 0$.  Otherwise, let $\mathcal{C}_d$ denote the linear space of functions that are constant on the support of die $d$. By Proposition~\ref{prop:nkequiv}, if ($\mathrm{NK}^+$) fails then there exists a die $d$ for which all functions $\psi \in \Ker \mathcal{D}$ also belong to $\mathcal{C}_d$, i.e., $\Ker \mathcal{D} \subseteq \mathcal{C}_d$.  Then $\mathcal{C}_d^\perp \subseteq (\Ker \mathcal{D})^\perp$, where $\perp$ indicates the dual subspace.  The space $(\Ker \mathcal{D})^\perp$ is the span of the probability mass functions $\mathrm{pmf}_d$ of all the dice.  Therefore every $\phi \in \mathcal{C}_d^\perp$ can be written as a linear combination
\[ \phi = \sum_{d \in \mathcal{D}} \beta(d) \cdot \mathrm{pmf}_d. \]
Then for every $\psi\colon \mathcal{F} \to \R$,
\[ \sum_{f \in \mathcal{F}} \phi(f) \cdot \psi(f) = \sum_{d \in \mathcal{D}, f \in \mathcal{F}} \beta(d) \cdot \mathrm{pmf}_d(f) \cdot \psi(f) = \sum_{d \in \mathcal{D}} \beta(d) \cdot \E_d[\psi(F)]. \]
The claim follows by specializing $\phi$ to the function that takes value $1$ on $f^*$, $-1$ on $f_*$, and 0 elsewhere.  This function is dual to $\mathcal{C}_d$.
\end{proof}

The $\mathrm{MVD}(\eps, \delta)$ (mean-variance divergence) condition postulates that there exists a function $\psi\colon \mathcal{F} \to [-1, 1]$ such that for every die $d \in \mathcal{D}$, 
\begin{align}\label{eq:E-V-epsilon-delta}
\big|\E_d[\psi(F)]\big| < \epsilon(\Var_d[\psi(F)] - \delta). \tag{MVD}
\end{align}

\begin{proposition}
\label{prop:mvd}
If GSV source $(\mathcal F, \mathcal D)$ fails ($\mathit{NK}^+$) then there exists a constant $C$ such that for every $\eps > 0$, $(\mathcal F, \mathcal D)$ fails $\mathit{MVD}(\eps, C \eps^2)$.
\end{proposition}
\begin{proof}
Assume $(\mathcal F, \mathcal D)$ fails ($\mathit{NK}^+$).  Let $d$ be the die stipulated by Claim~\ref{claim:nkdual} and $C$ be the maximum of $(\sum_{d' \in \mathcal{D}} \abs{\beta(d')})^2$ over all pairs of faces $f^*, f_*$ in the support of $d$.

Towards a contradiction suppose that $(\mathcal F, \mathcal D)$ satisfies $\mathit{MVD}(\eps, \delta)$.  Then the witness $\psi\colon \mathcal{F} \to [-1, 1]$ for $\mathit{MVD}(\eps, \delta)$ must satisfy the conditions $\Var_d[\psi(F)] > \delta$ and $\abs{\E_{d'}[\psi(F)]} < \eps \Var_{d'}[\psi(F)]$ for all dice $d' \in \mathcal{D}$.  Let $f^*$ and $f_*$ be faces in the support of $d$ that maximize and minimize the value of $\psi$, respectively.  By Claim~\ref{claim:nkdual}, relation \eqref{eq:nkpdual} holds for some $\beta$ that may depend on $f^*$ and $f_*$ but not on $\psi$.  Then
\begin{multline*}
\sqrt{\delta} 
  < \sqrt{\Var_d[\psi(F)]} 
  \leq \psi(f^*) - \psi(f_*) 
  = \sum_{d' \in \mathcal{D}} \beta(d') \cdot \E_{d'}[\psi(F)] \\
  \leq \sum_{d' \in \mathcal{D}} \abs{\beta(d')} \cdot \abs{\E_{d'}[\psi(F)]} 
  < \sum_{d' \in \mathcal{D}} \abs{\beta(d')} \cdot \eps \Var_{d'}[\psi(F)] 
  \leq \sqrt{C} \eps, 
\end{multline*}
where the last inequality follows from the definition of $C$ and the boundedness of $\psi$.  Therefore $\mathit{MVD}(\eps, \delta)$ fails for $\delta = C \eps^2$.
\end{proof}

\begin{proposition}
\label{prop:lowerquality}
Assume that $(\mathcal F, \mathcal D)$ fails $\mathit{MVD}(\eps, \delta)$.  Then every extractor with error $\eps/20$ for $(\mathcal F, \mathcal D)$ requires $1/8\delta$ samples, assuming $\eps > 0$ is sufficiently small. 
\end{proposition}
\begin{proof}
The proof is a direct extension of the proof of Proposition~\ref{prop:necessary}.  The main technical tool is the following claim:

\begin{claim}
\label{claim:ext-induction}
For every extractor $Ext\colon \mathcal{F}^n \to \{0, 1\}$, and every $0 \leq \alpha \leq 1$, if $\E_{A_-}[Ext] \geq \alpha$ for every strategy $A_-$, then there exists a strategy $A_+$ for which 
\[ \E_{A_+}[Ext] \geq \alpha + \frac{\eps}{1 + \eps} \cdot \bigl(\alpha (1 - \alpha) - \delta n\bigr). \]
\end{claim}

The proof of Claim~\ref{claim:ext-induction} is a notationally intensive direct extension of the proof of Claim~\ref{claim:induction}.  We omit the details.  

By Claim~\ref{claim:ext-induction} it follows that for every $\eps > 0$, if no strategy $A_-$ has error less than $\alpha = 1/2 - \eps/20$ against $Ext$ then there exists a strategy $A_+$ with advantage at least
\[ \E[Ext] \geq \frac12 - \frac{\eps}{20} + \frac{\eps}{1 + \eps} \cdot \Bigl(\frac{1 - \eps^2/400}{4} - \frac{1}{8}\Bigr), \]
which is at least $1/2 + \eps/20$ for sufficiently small $\eps$.
\end{proof}

\section{Open Questions}

In this work, we completely classify GSV sources in terms of their extractability.  We point out the following questions for further investigation:

\begin{itemize}
\item Is the sample complexity of $o(1/\eps^2)$ in part 2 of Theorem~\ref{thm:main2} tight?  Example E2 gives an upper bound of $O(1/\eps^7)$.  This non-$\mathrm{NK}^+$ extractable source satisfies $\mathrm{MVR}(\eps)$ with minimum variance $\eps^2$ for every $\eps$ (with witness $\psi = (\eps, -\eps, 1, -1)$), so $O(1/\eps^7)$ are sufficient for extraction error $\eps$ by Proposition~\ref{prop:sufficient}.

\item The number of required samples in Theorem~\ref{thm:main1} is of the form $\eps^{-O(2^{\abs{\mathcal{D}}})}$, where $\abs{\mathcal{D}}$ is the number of dice (see the proof of Proposition~\ref{prop:char}).  Is this exponential dependence in $\abs{\mathcal{D}}$ necessary?

\item  The multi-bit extractor in Theorem~\ref{thm:main2} runs in time $\min(nm2^{m},n^{O(\abs{\mathcal{F}})})$.  Can the dependence on the number of faces be improved, possibly by applying known seeded extraction algorithms?  

\item Proposition~\ref{prop:sufficient} states that sources satisfying condition $\mathrm{MVR}(\eps)$ admit extraction with error $O(\sqrt{\eps})$, while by Proposition~\ref{prop:necessary} extraction error $\Omega(\eps)$ is necessary.  Can this quadratic gap be narrowed?
\end{itemize}

\subsection*{Acknowledgments}
Part of this work was done while Omid Etesami and Siyao Guo were visiting the Chinese University of Hong Kong, and while Andrej Bogdanov and Siyao Guo were visiting the Simons Institute for the Theory of Computing at UC Berkeley.

\bibliographystyle{alpha}
\bibliography{refs}
\end{document}